\documentclass[preprint,12pt]{elsarticle}
\usepackage{enumerate,amssymb,amsmath,amsthm,mathrsfs,latexsym,amsfonts,txfonts,graphicx,multicol,color,mathrsfs}
\usepackage{multirow,float}

\theoremstyle{plain}
\newtheorem{thm}{Theorem}[section]
\newtheorem{lem}[thm]{Lemma}

\newtheorem{cor}[thm]{Corollary}
\theoremstyle{definition}

\newtheorem{prop}[thm]{Proposition}

\newtheorem{rem}[thm]{Remark}
\newcommand{\CH}{\ensuremath{C^{{\perp}_H}}}

\newcommand{\Fqq}{\ensuremath{\mathbb{F}_{q^2}}}
\newcommand{\Fq}{\ensuremath{\mathbb{F}_{q}}}

\newcommand{\be}{\ensuremath{\beta}}
\newcommand{\gam}{\ensuremath{\gamma}}
\newcommand{\nb}{\ensuremath{\overline{n}}}
\newcommand{\B}{\ensuremath{B_{{\nb,\lambda,q^2}}^H}}
\DeclareMathOperator{\lcm}{lcm}

\DeclareMathOperator{\ord}{ord}
\DeclareMathOperator{\Hull}{Hull}

\journal{ }

\begin{document}
\begin{frontmatter}
\title{The Average Dimension of the Hermitian Hull of Constayclic Codes over Finite Fields}
 \author[label1]{Somphong Jitman} 
   \address[label1]{Department of Mathematics, Faculty of Science, Silpakorn University, Nakhon Pathom 73000, Thailand}

  \address[label2]{Department of Mathematics and Statistics, Faculty of Science, Thaksin University, Phattalung 93110, Thailand}

 \author[label2]{Ekkasit Sangwisut\footnote{Corresponding author.\\ E-mail address:ekkasit.sangwisut@gmail.com (Ekkasit Sangwisut).}}



\begin{abstract}  
    The hulls of  linear  and cyclic codes  have been extensively studied due to their  wide applications. 
    The dimensions and  average dimension of the Euclidean hull of linear  and cyclic codes have been well-studied. 
   In this paper, the average dimension of the Hermitian hull  of constacyclic codes of length $n$ over a finite field $\mathbb{F}_{q^2}$  is  determined  together with some upper and lower bounds. It turns out that   either the average dimension of the Hermitian hull  of constacyclic codes of length $n$ over $\mathbb{F}_{q^2}$ is  zero or  it  grows the same rate as $n$.   Comparison  to the average dimension of  the  Euclidean  hull of cyclic codes  is  discussed as well.
\end{abstract}

\end{frontmatter}

\thispagestyle{empty}
\section{INTRODUCTION}
The (Euclidean) hull of a  linear code has been introduced to classify finite projective planes in \cite{Assmus90}. It is defined to be the intersection of a linear code and its Euclidean dual.  Later, it turns out that the  hulls of linear codes play a vital  role in determining the complexity of   algorithms for  checking permutation equivalence of two linear codes  in \cite{Leon91, Sendrier997, Sendrier00}. Subsequently, it has been shown that  the hull is an indicator for  the complexity of   algorithms       for   computing   the automorphism group of a linear code in \cite{Leon82, Sendrier01}.  Precisely, most of the algorithms do not work if the size of the hull  is large.  Recently, the hulls of linear codes have been applied in constructing  good entanglement-assisted quantum error correcting codes in \cite{GJG2016}.
Due to these  wide applications, the hulls of linear codes  and their properties  have been extensively studied.  The number of linear codes of length $n$ over $\Fq$ whose hulls have  a common dimension and the average dimension of the hull of linear codes were studied in \cite{Sendrier97}. It has been shown that the average dimension of the hull of linear codes   is asymptotically a positive constant dependent of $q$.

Constacyclic codes constitute an  important class of linear codes  due to their nice algebraic structures and various applications in engineering \cite{Bakshi2012} and \cite{Chen2012}.  Especially,  this family of codes contains   a class of well-studied cyclic codes. In \cite{Sangwisut}, the number of cyclic codes of length $n$ over $\Fq$ having hull of a fixed dimension has been determined together with the  dimensions of the hulls of cyclic codes of length $n$ over $\Fq$.
The average dimension of the hull of cyclic codes   with respect to the Euclidean and Hermitian inner products have been  investigated in \cite{Skersys03} and \cite{Sangwisut16}, respectively. It has been shown that either  the  average dimension of the hull of such codes is zero or it grows at the same rate with $n$.  In \cite{Sangwisutccqtc}, the dimensions of the Hermitian hulls of {constacyclic} codes of length $n$ over $\Fqq$ have been determined. However, in the literature, the average dimension of the Hermitian hull of constacyclic codes has not been studied. Therefore, it is of natural interest to study the average dimension of the Hermitian hull of constacyclic codes.

 In this paper,    we focus on the average dimension of the Hermitian hull of constacyclic codes.   Employ the techniques modified from \cite{Sangwisutccqtc} and \cite{Skersys03}, a general formula for  the  average dimension of the Hermitian hull of constacyclic  codes  of length $n$ over $\Fqq$ is determined. 
Asymptotically,   either the  average dimension of the Hermitian hull of constacyclic  codes is    zero or it grows at the same rate with $n$. This result coincides with the case of the average dimension of the  Euclidean hull of cyclic codes.  However,   there are interesting  differences on  lower bounds discussed in Section 6.

The  paper is organized as follows.  In Section 2, some basic knowledge concerning polynomials and codes over finite fields are recalled.  A  general formula for the average dimension of the Hermitian hull of constacyclic codes is given in Section 3. In Section 4,  some number theoretical tools are discussed  together with  a simplified formula for  the average dimension of the Hermitian hull of constacyclic codes.  Lower and upper bounds for  the average dimension of the Hermitian hull of constacyclic codes are studied in Section 5. The summary and remarks are given in Section 6.

\section{PRELIMINARIES}	

 The main focus of this paper is the average dimension of  the Hermitian hull of constacyclic codes  which is  well-defined only over a  finite field of square order (see Equation \eqref{eq:hull}).  In this section,  some basic properties of codes and polynomials over such  finite fields.
 
 For convenience, let $p$ be a prime,  $q$ be a $p$-power integer and let 
$\Fqq$ denote a finite field of order $q^2 $ and  characteristic $p$.    For  a given  positive  integer $n$, let $\Fqq^n$ denote the $\Fqq$-vector space of all vectors of length $n$ over $\Fqq$.   For $0\leq k \leq n$, a \textit{linear code} of length $n$ and dimension $k$ over $\Fqq$ is defined to be a $k$-dimensional subspace of the $\Fqq$-vector space $\Fqq^n$.  The \textit{Hermitian dual} of a linear code $C$ is defined to be the set
$$\CH:=\left\{\mathbf{u}\in\Fqq^n\mid \mathbf{u}\cdot \mathbf{v}^q=0\text{~~for all~~} \mathbf{v}\in C\right\},$$
where $\mathbf{u}\cdot \mathbf{v}^q:=\sum_{i=0}^{n-1}u_i\cdot v_i^q$ for all $\mathbf{u}=(u_0, u_1,\ldots, u_{n-1})$ and $\mathbf{v}=(v_0, v_1, \ldots, v_{n-1})$ in $\Fqq^n$. The \textit{Hermitian hull} of a linear code $C$ is defined to be \begin{align}
\label{eq:hull}
\Hull_H(C):=C\cap \CH.
\end{align}

For a fixed  nonzero element $\lambda$ in $\Fqq$, a linear code of length $n$ over $\Fqq$ is said to be {\em constacyclic}, or specifically, {\em $\lambda$-constacyclic} if  $(\lambda c_{n-1},c_1,\ldots,c_{n-2})\in C$ whenever  $(c_0,c_1,\ldots,c_{n-1})\in C$.  Every  $\lambda$-constacyclic code $C$ of length $n$ over $\Fqq$ can be   identified by  an  ideal in the principal ideal ring $\Fqq[x]/\langle x^n-\lambda\rangle$   uniquely  generated by a   monic divisor of $x^n-\lambda$. In this case, In this case,  $g(x)$  is called the \textit{generator polynomial} for  $C$ and  we have $\dim C= n-\deg(g(x)$.

For each  polynomial $f(x)=a_0+a_1x+\ldots+a_kx^k\in\Fqq[x]$ of degree $k$ and $a_0\neq 0$, the \textit{conjugate-reciprocal polynomial} of $f(x)$ is  defined to be  $f^\dagger(x):=a_0^{-q}\sum_{i=0}^k a_i^qx^{k-i}$. It is not difficult to see that $\left(f^\dagger\right)^\dagger(x)=f(x)$.   A polynomial $f(x)$ is said to be  \textit{self-conjugate-reciprocal} if $f(x)=f^\dagger(x)$. Otherwise, $f(x)$ and $f^\dagger(x)$ are called a \textit{conjugate-reciprocal polynomial pair}.

Denote by  $r$  the order of an element $\lambda$ in the multiplication group $\Fqq^*:=\Fqq\setminus\{0\}$.   In \cite[Proposition 2.3]{Yang13}, it has been shown that the Hermitian dual of a $\lambda$-constacyclic code over $\Fqq$ is again  $\lambda$-constacyclic if and only if $r |(q+1)$.  Based on this characterization, we assume  that $\lambda$ is an element in  $\Fqq$  of order $r$ such that  $r | (q+1)$ throughout the paper. 

Let $C$ be a  $\lambda$-constacyclic code of length $n$ over $\Fqq$ with the generator polynomial $g(x)$ and let  $h(x)=\frac{x^n-\lambda}{g(x)}$.  Then  $h^\dagger(x)$ is a monic divisor of $x^n-\lambda$ and it is the generator polynomial of $C^{\perp_H}$(see \cite[Lemma 2.1]{Yang13}). By \cite[Theorem 1]{Sangwisut},  the hull $\Hull_H(C)$ of $C$ is generated by the polynomial $\lcm(g(x),h^\dagger(x))$.

Let $\mathcal{C}(n,\lambda,q^2)$ denote the {set of all $\lambda$-constacyclic codes of length $n$ over $\Fqq$}. The {\em average dimension} of the Hermitian hull of $\lambda$-constacyclic codes of length $n$ over $\Fqq$  is defined to be 
$$E_H(n,\lambda,q^2):=\sum_{C\in \mathcal{C}(n,\lambda,q^2)}\frac{\dim \Hull_H(C)}{|\mathcal{C}(n,\lambda,q^2)|}.$$

  For each positive integer $n$, it  can be written in the form of $n=\nb p^\nu$, where $p\nmid \nb$ and $\nu\geq 0$.   Since the map $\alpha\mapsto \alpha^{p^\nu}$ is an automorphism on $\Fqq$, there exists an element 
  $\Lambda\in \Fqq$ such that  $\Lambda^{p^\nu}=\lambda$ and the multiplicative order of $\Lambda$ is $r$.  
Using arguments  similar to those in \cite[Section 3]{Sangwisutccqtc},  up to permutation, there exist nonnegative integers $\mathtt{s}$ and $\mathtt{t}$ such that 
\begin{align}\label{xnlam}
x^n-\lambda=\left(x^{\nb}-\Lambda\right)^{p^\nu}
=&~~\prod_{i=1}^{\mathtt{s}}\left(g_i(x)\right)^{p^{\nu}}\prod_{j=1}^{\mathtt{t}}\left(f_j(x)\right)^{p^\nu}\left(f_j^\dagger(x)\right)^{p^\nu}, 
\end{align}
where {$f_j(x)$ and $f_j^\dagger (x)$ are a conjugate-reciprocal polynomial pair} and   $g_i(x)$ is a monic irreducible self-conjugate-reciprocal polynomial  for all $1\leq i\leq \mathtt{s}$ and $1\leq j\leq \mathtt{t}$. 

Based on   the factorization in Equation \eqref{xnlam},    the  generator polynomial of  a  $\lambda$-constacyclic code  $C$ of length $n$ over $\Fqq$   can be viewed  of the form \[\displaystyle g(x)=\prod_{i=1}^{\mathtt{s}}g_i(x)^{u_i}\prod_{j=1}^{\mathtt{t}}f_j(x)^{z_j}\left(f^\dagger_j(x)\right)^{w_j},\]
where $0\leq u_i, z_j, w_j\leq p^{\nu}$. It follows that the  generator polynomial of $\CH$ is  $$\displaystyle h^\dagger(x)=\prod_{i=1}^{\mathtt{s}}g_i(x)^{p^{\nu}-u_i}\prod_{j=1}^{\mathtt{t}}f_j(x)^{p^{\nu}-w_j}\left(f^\dagger_j(x)\right)^{p^{\nu}-z_j},$$ and hence, the generator polynomial of  $\Hull_H(C)$ is 
\begin{align}\label{lcm} 
\lcm(g(x), h^\dagger(x))=\prod_{i=1}^{\mathtt{s}}g_i(x)^{\max\{u_i,~ p^{\nu}-u_i\}}\prod_{j=1}^{\mathtt{t}}f_j(x)^{\max\{z_j,~ p^{\nu}-w_j\}}\left(f^\dagger_j(x)\right)^{\max\{w_j. p^{\nu}-z_j\}}.\end{align}
It follows that \begin{align} \label{eq:dimhull} \dim Hull_H(C)
=&~ n-\sum_{i=1}^{\mathtt{s}}\deg(g_i(x))\left(\max\{u_i,p^\nu-u_i\}\right) \notag \\
&~ -\sum_{j=1}^{\mathtt{t}}\deg(f_j(x))\left(\max\{z_j,p^\nu-w_j\}+\max\{w_j,p^\nu-z_j\}\right)
\end{align}

\section{The Average Dimension     $
    E_H(n,\lambda, q^2)$}
In this section, a general formula for the average dimension  $
E_H(n,\lambda, q^2)$  of the Hermitian hull of $\lambda $-constacyclic codes of length $n$ over $\Fqq$  is given together with some  upper bounds.

Assume that $x^{\overline{n}}-\Lambda$ has the  factorization  in the form of Equation \eqref{xnlam} and  let   $B_{{\nb,\lambda,q^2}}^H:=\sum\limits_{i=1}^\mathtt{s} \deg(g_i(x))$. 
The formula for the average dimension of the Hermitian hull of constacyclic codes can be   determined using the expectation $E(\,\cdot\,) $ in  Probability Theory   as follows.
\begin{lem}\label{tu}
	Let $p$ be a prime and let $\nu$ be a nonnegative integer. Let $0\leq u, z, w\leq p^\nu$. Then the following statements hold.
	\begin{enumerate}
		\item $E(\max\{u,p^\nu-u\})=\frac{3p^\nu+1}{4}-\frac{\delta_{p^\nu}}{4(p^\nu+1)}, $
        where $\delta_{p^\nu}=1$ if $p^\nu$ is even, and $\delta_{p^\nu}=0$ otherwise.
		\item $E(\max\{z,p^\nu-w\})=\frac{p^\nu(4p^\nu+5)}{6(p^\nu+1)}$.
	\end{enumerate}
\end{lem}
\begin{proof} The statements can be obtained using arguments similar to those in  the proof of 
	\cite[Theorem 23]{Skersys03}.
\end{proof}

\begin{thm}\label{thm1}
    Let $\Fqq$ be a finite field of order $q^2$ and characteristic $p$ and let $n$ be a positive integer such that $n=\nb p^{\nu}, p \nmid \nb$ and $\nu\geq 0$.  Let $\lambda$  be an element in $\Fqq$ of order $r$ such that $r|(q+1)$. Then the average dimension of the Hermitian hull of $\lambda$-constacyclic codes of length $n$ over $\Fqq$ is
    \begin{align}\label{EH}
    E_H(n,\lambda, q^2)=n\left(\frac{1}{3}-\frac{1}{6(p^\nu+1)}\right)-B_{{\nb,\lambda,q^2}}^H\left(\frac{p^\nu+1}{12}+\frac{2-3\delta_{p^\nu}}{12(p^\nu+1)}\right). 
    \end{align}	   
\end{thm}	
\begin{proof}

 Let $Y$ be  the random variable  of the dimension  $\dim Hull_{H}(C)$,  where $C$  is chosen randomly from $\mathcal{C}(n,\lambda,q^2)$ with uniform probability. By Lemma~\ref{tu}, Equation \eqref{eq:dimhull}, and  arguments similar to those in the proof of \cite[Theorem 3.2]{Sangwisut16}, we obtain
\begin{align*}
E_H(n,\lambda, q^2)=&E(Y)\\
=&n-\sum_{i=1}^s\deg(g_i(x))E\left(\max\{u_i,p^\nu-u_i\}\right)\\
&~ -\sum_{j=1}^t\deg(f_j(x))E\left(\max\{z_j,p^\nu-w_j\}+\max\{w_j,p^\nu-z_j\}\right)\\
=&n- \sum_{i=1}^s\deg(g_i(x)) \left(\frac{3p^\nu+1}{4}-\frac{\delta_{p^\nu}}{4(p^\nu+1)}\right)\\
&~ - \sum_{j=1}^t\deg(f_j(x)) \frac{2p^\nu(4p^\nu+5)}{6(p^\nu+1)}\\
=&n-\left(\frac{3p^\nu+1}{4}-\frac{\delta_{p^\nu}}{4(p^\nu+1)}\right)B_{{\nb,\lambda,q^2}}^H-\frac{p^\nu(4p^\nu+5)}{6(p^\nu+1)}\left(\nb-B_{{\nb,\lambda,q^2}}^H\right)\\
=&n\left(\frac{1}{3}-\frac{1}{6(p^\nu+1)}\right)-B_{{\nb,\lambda,q^2}}^H\left(\frac{p^\nu+1}{12}+\frac{2-3\delta_{p^\nu}}{12(p^\nu+1)}\right).
\end{align*}

The proof is  therefore completed.
\end{proof}

The following corollary is straightforward from Theorem~\ref{thm1}. 

\begin{cor}\label{coreq}  
    Assume the notations as in Theorem~\ref{thm1}. Then the  following statements hold.
    \begin{enumerate}
        \item $E_H(n,\lambda,q^2)<\frac{n}{3}$.
        \item $E_H(\nb,\lambda,q^2)=\frac{\nb-B_{{\nb,\lambda,q^2}}^H}{4}$.
        \item $E_H(\nb,\lambda,q^2){<}\frac{\nb}{4}$. 
    \end{enumerate}
\end{cor}

\section{Properties of $B_{\nb,\lambda,q^2}^H$}

{ 
    
    In this section,  some properties of  $B_{\nb,\lambda,q^2}^H$ are  given as well as their applications in determining a simplified formula   for $E_H(n,\lambda,q^2)$.

Let	$M_q:=\{\ell\geq 1| \ell \text{~divides~} q^i+1 \text{~for some odd positive integer~} i\}$ and let \begin{align} \chi:=\left\{j\geq 1~\middle|~ j|\nb r \text{ and }\gcd\left(\frac{\nb r}{j},r\right)=1\right\}.
\label{eq:chi}\end{align}
For each positive integer $j$ such that $\gcd(j,q)=1$, denote by $\ord_j(q)$ the {\em multiplicative order} of $q$ modulo $j$. 

The formula for $B_{{\nb,\lambda,q^2}}^H$ can be simplify using the sets $M_q$ and $\chi$ as follows.

\begin{lem}\label{BB}
    Assume that $x^{\overline{n}}- \Lambda$ is factorized as in  Equation \eqref{xnlam}. Then
    \[B_{{\nb,\lambda,q^2}}^H= \displaystyle\sum_{j\in \chi\cap M_q}\frac{\phi(j)}{\phi(r)},\] where $\phi$ is the Euler's totient function.
\end{lem}
\begin{proof}
    By Equation \eqref{xnlam}, we have  \[x^{\nb}-\Lambda 
    =~\prod_{i=1}^{\mathtt{s}}g_i(x) \prod_{j=1}^{\mathtt{t}}f_j(x)f_j^\dagger(x).\]  From \cite[Equation (3.11)]{Sangwisutccqtc}, $x^{\nb}-\Lambda$ can be factored as  
    \begin{align*}
    x^{\nb}-\Lambda =  \prod_{j\in\chi\cap M_q}\prod_{i=1}^{\gam(j)}g_{ij}(x)\prod_{j\in\chi\smallsetminus M_q}\prod_{i=1}^{\be(j)}f_{ij}(x)f^\dagger_{ij}(x),
    \end{align*}
    where
    $\gam(j)=\frac {\phi(j)}{\phi(r)ord_j(q^2)}$, $\label{bet}
    \be(j)=\frac {\phi(j)}{2\phi(r)ord_j(q^2)}$,  $f_{ij}(x)$ and $f^\dagger_{ij}(x)$ are a monic irreducible conjugate-reciprocal polynomial pair of degree $ord_j(q^2)$, 	 and ~$g_{ij}(x)$~ is a monic irreducible self-conjugate-reciprocal polynomial of degree $ord_j(q^2)$.

    Altogether, it can be  concluded that 
    $$\prod_{i=1}^{\mathtt{s}}g_i(x)=\prod_{j\in\chi\cap M_q}\prod_{i=1}^{\gam(j)}g_{ij}(x).$$
    Hence,
    \begin{align*}
    B_{{\nb,\lambda,q^2}}^H&=\sum\limits_{i=1}^\mathtt{s} \deg(g_i(x))\\
    & =\sum_{j\in\chi\cap M_q}\gam(j)\deg(g_{ij}(x))\\
    &=\sum_{j\in\chi\cap M_q}\frac {\phi(j)}{\phi(r)ord_j(q^2)}\cdot ord_j(q^2)\\
    &=\sum_{j\in \chi\cap M_q}\frac{\phi(j)}{\phi(r)}
    \end{align*}
    as desired.
\end{proof}
\begin{rem}\label{B=0} From Lemma \ref{BB},  we have the following facts.
    \begin{enumerate}
        \item  If $\lambda=1$, then $B_{\nb, 1,q^2}^H$ is alway  positive since $1\in\chi\cap M_q$. 
        \item  If $\lambda\neq 1$, then  $\chi\cap M_q$ can be empty. In this case, $\B=0$. For example,   $B_{4,2,9}^H=0$  since $\chi\cap M_3=\emptyset$.
    \end{enumerate}   
\end{rem}

Recall that $\lambda$ is an element in   $\Fqq$ of  order $r$ such that  $r|(q+1)$.  Assume that the  prime factorization of $\gcd(\nb,r)$ is of the form 
$$\gcd(\nb,r)=2^{c_0}p_1^{c_1}\ldots p_s^{c_s}$$ for some $s\geq 0$,  where $p_1, p_2,\dots, p_s$ are distinct odd primes,  $c_0\geq 0$ and $c_i\geq 1$ for all $1\leq i\leq s$.  Then $\nb$ and $r$  can be factorized  in the forms of 
\begin{equation}\label{nbr}
	\nb=2^{\be(\nb)}p_1^{a_1}\ldots p_s^{a_s}\mu \text{~~and~~} r=2^{\be(r)}p_1^{b_1}\ldots p_s^{b_s}\tau,
\end{equation}
where $\be(\nb)\geq c_0$ and $\be(r)\geq c_0$  are integers, $\mu$ and $\tau$ are odd   (not necessarily prime)  integers relative prime to $p_i$ for all $1\leq i\leq s$, $a_i$ and $b_i$ are positive integers.

Based on the factorizations above,      the presentation of the  set $\chi$  can be simplified  as follows.

 \begin{lem}\label{lem4.7}
 Let $\nb$ and $ r$ be positive integers and let  $\chi$  be defined as in Equation \eqref{eq:chi}. Then \[\chi=
 \begin{cases}
\left\{r2^{\be(\nb)}\displaystyle\prod_{i=1}^s p_i^{a_i}k ~\middle|~ k| \mu\right\}& \text{  if $r$ is even,}\\
~\\
 \left\{r\displaystyle\prod_{i=1}^s p_i^{a_i}k ~\middle|~ k| 2^{\be(\nb)}\mu\right\} &\text{ if $r$ is odd.}
 \end{cases}\]
\end{lem}
\begin{proof}
Consider the following $2$ cases.     

\noindent \textbf{Case 1} $r$ is even.
We have
\begin{align}\notag
\chi=&~~\left\{j\geq 1~\middle|~ j|\nb r\text{ and }  \gcd\left(\frac{\nb r}{j},r\right)=1\right\} \notag\\
=&~~\left\{2^{\be(\nb)+\be(r)}p_1^{a_1+b_1}\ldots p_s^{a_s+b_s}\tau k ~\middle|~ k| \mu\right\} \notag \\ \label{ome}
=&~~\left\{2^{\be(\nb)+\be(r)}\prod_{i=1}^s p_i^{a_i+b_i}\tau k ~\middle|~ k| \mu\right\}\\\label{omega}
=&~~\left\{r2^{\be(\nb)}\prod_{i=1}^s p_i^{a_i}k ~\middle|~ k| \mu\right\}.
\end{align}

\noindent \textbf{Case 2} $r$ is odd.  We have 
\begin{align}\notag
\chi=&~~\left\{j|\nb r~\middle|~ \gcd\left(\frac{\nb r}{j},r\right)=1\right\} \notag \\ 
=&~~
\left\{p_1^{a_1+b_1}\ldots p_s^{a_s+b_s}\tau k ~\middle|~ k| 2^{\be(\nb)}\mu\right\}\notag \\ \label{ome7}
=&~~\left\{\prod_{i=1}^s p_i^{a_i+b_i}\tau k ~\middle|~ k| 2^{\be(\nb)}\mu\right\}\\\label{omega8}
=&~~\left\{r\prod_{i=1}^s p_i^{a_i}k ~\middle|~ k| 2^{\be(\nb)}\mu\right\}.
\end{align}
Combining the two cases, the result follows.
\end{proof}

For   integers $i\geq 0$ and  $j\geq 1$, we say that $2^i$ exactly divides $j$, denoted by $2^i || j$, if $2^i$ divides $j$ but $2^{i+1}$ does not divide $j$.

From 
\cite[Corollary 3.7]{Jitman14} and its proof,  we have the following proposition.
\begin{prop}\label{ell}
    Let $q$ be a prime power and let $\ell=2^\be\overline{\ell}$ be a positive integer such that $\overline{\ell}$ is odd and $\be\geq 0$. Let $\gam\geq 0$ be an integer such that $2^\gam||(q+1)$. If $\ell\in M_q$, then one of the following statements holds.
    \begin{enumerate}
        \item $\ell=1$ and $q$ is even.
        \item $\ell\in\{1, 2\}$ and $q$ is odd.
        \item $\ell>2$ and one of the following statements holds. 
        \begin{enumerate}
            \item $\overline{\ell}=1$ and $2\leq\be\leq\gam$.
            \item $\overline{\ell}\geq 3$, $2||\ord_{p}(q)$  for every prime $p$ dividing $\overline{\ell}$, and $0\leq\be\leq\gam$.
        \end{enumerate}
    \end{enumerate}
\end{prop}    
} The necessary and sufficient conditions for   $\B$ to be non-zero are given   as follows.
\begin{lem}\label{lem48}
	Let  $\gam\geq 0$ be such that $2^\gam || (q+1)$. Then one of the following statements holds.
	\begin{enumerate}
		\item  If $r$ is even, then	$\B\neq 0$ if and only if $\be(\nb)+\be(r)\leq \gam$ and $r\in M_q$.
		\item If  $r$ is odd, then	$\B\neq 0$ if and only if $r\in M_q$.
	\end{enumerate}
	
\end{lem}
\begin{proof}
From Lemma~\ref{BB}, it is not difficult  to see that $\B\neq 0$ equivalent to $\chi\cap M_q\neq\emptyset$.

To prove $1)$, assume that $r$ is even.  Suppose that $\B\neq 0$. Then $\chi\cap M_q\ne \emptyset$.  Let $j$ be  an element in $\chi\cap M_q$.  By Equations \eqref{ome} and \eqref{omega},  we have $$ j=2^{\be(\nb)+\be(r)}\prod_{i=1}^s p_i^{a_i+b_i}\tau k=r2^{\be(\nb)}\prod_{i=1}^s p_i^{a_i}k$$ for some $k|\mu$.  { Since $j\in M_q$, we have $  \be(\nb)+\be(r)\leq\gam$ and $r\in M_q$ by Proposition~\ref{ell}.}

Conversely, assume that $\be(\nb)+\be(r)\leq \gam$ and $r\in M_q$. By setting $k=1$,   we have 
$$j=2^{\be(\nb)+\be(r)}\prod_{i=1}^s p_i^{a_i+b_i}\tau =r2^{\be(\nb)}\prod_{i=1}^s p_i^{a_i}\in\chi\cap M_q$$
since  $r$ is even.  

To  prove $2)$, assume that $r$ is odd. Assume that $\B\neq 0$. Then   $\chi\cap M_q\neq \emptyset$.  Let $j$ be  an element  in $\chi\cap M_q$.  Then 
\begin{align*}
j=\prod_{i=1}^sp_i^{a_i+b_i}\tau k=r\prod_{i=1}^sp_i^{a_i}k
\end{align*}
for some $k| 2^{\be(\nb)}\mu$ by Equations \eqref{ome7} and \eqref{omega8}. { Since $j\in M_q$ and $r| j$,  we have $r\in M_q$ by Proposition~\ref{ell}.}

Conversely,  assume that $r\in M_q$.  By setting $k=1$, we have
\begin{align*}
j=\prod_{i=1}^sp_i^{a_i+b_i}\tau=r\prod_{i=1}^sp_i^{a_i}\in\chi\cap M_q
\end{align*}
since  $r$ is odd. 
\end{proof}

The next lemma can be deduced form  Lemma \ref{lem48}.
\begin{cor}\label{b=0}
	Let    $\gam\geq 0$ be such that $2^\gam\parallel (q+1)$.  Then one of the following statements holds.
	\begin{enumerate}
		\item 	If  $r$ is even, then	$B_{\nb,\lambda,q^2}=0$ if and only if $\be(\nb)+\be(r)> \gam$ or $r\not\in M_q$.
		\item If  $r$ is odd, then $B_{\nb,\lambda,q^2}=0$ if and only if $r\not\in M_q$.
	\end{enumerate}
\end{cor}

In the case where $B_{\nb,\lambda,q^2}=0$, the average dimension $E_H(n,\lambda,q^2)$ can be  simplified  from Theorem ~\ref{thm1} and Corollary ~\ref{b=0} as follows.
{ \begin{cor}\label{cor-noB}
	Let $q$ be a power of a prime $p$ and let  $n\geq 1$. Then the following statements holds.
	\begin{enumerate}
		\item If $r$ is even, then $E_H(n,\lambda,q^2)=n\left(\frac{1}{3}-\frac{1}{6(p^\nu+1)}\right)$ if and only if $\be(\nb)+\be(r)>\gam$ or $r\not\in M_q$.
		\item If  $r$ is odd, then $E_H(n,\lambda,q^2)=n\left(\frac{1}{3}-\frac{1}{6(p^\nu+1)}\right)$ if and only if $r\not\in M_q$.
	\end{enumerate}
\end{cor}

Let $\ell$ be a positive integer relatively prime  to $q$. Let $\ell=2^\be p_1^{e_1}\ldots p_k^{e_k}$ be the prime factorization of $\ell$ where $e_i\geq 1,~ \be\geq 0$ and $k\geq 0$.  Let $K^\prime:=\{i| p_i\not\in M_q \} $ and $K_1:=\{i| p_i\in M_q\}$. Then $K^\prime$ and $K_1$ form a partition of $\left\{1, \ldots, k\right\}$.
Let $d'(\ell)=\prod\limits_{i\in K'}p_i^{e_i}$ and $d_1(\ell)=\prod\limits_{i\in K_1}p_i^{e_i}$. For convenience, the empty product will be regards  as $ 1$. We therefore have  $\ell=2^{\be(\ell)} d'(\ell)d_1(\ell)$ and it is  called the $M_q$-\textit{factorization} of~$\ell$. 

From  Corollary~\ref{b=0},   if $r\not\in M_q$, then $B_{\nb,\lambda,q^2}=0$. Next, the expression of $\B$  when $r\in M_q$ is determined  using properties of set $M_q$ and the expression of $\chi$ in the following proposition.  By  Proposition~\ref{ell}, $r$ is of the form $r=2^{\be(r)}d_1(r)$. Since $r|(q+1)$, we have $2^{\be(r)}| (q+1)$.
\begin{prop}\label{bnbq}
  Let $\nb=2^{\be(\nb)}d'(\nb)d_1(\nb)$ and $r=2^{\be(r)}d_1(r)$ be  the $M_q$-factorizations of $\nb$ and $r$.  Then one of the following statements holds.
	\begin{enumerate}
		\item If  $r$ is even,~ $\be(\nb)+\be(r)\leq \gam$, and $r\in M_q$, then $B_{\nb,\lambda,q^2}^H=2^{\be(\nb)}d_1(\nb)$.
		\item If  $r$ is odd and  $r\in M_q$, then $\B=2^{\min\{\gam,~ \be{(\nb)}\}}d_1(\nb)$.
	\end{enumerate}
\end{prop}
\begin{proof}
Assume that  $r$ is even, $\be(\nb)+\be(\gam)\leq\gam$, and $r\in M_q$.	We consider the following 2 cases.
	
\noindent \textbf{Case 1}  $\gcd(\nb,r)=1$. 

\noindent \textbf{Case 1.1}  $\be(\nb)=0$. Then $\nb=d'(\nb)d_1(\nb)$ and $r=2^{\be(r)}d_1(r)$. By Lemma~\ref{BB}, it can be concluded that
\begin{align*}
\B=\frac{1}{\phi(r)}\sum_{j\in\chi\cap M_q}\phi(j)=\frac{1}{\phi(r)}\sum_{k|\nb,~ kr\in M_q}\phi(kr)=\sum_{k| d_1(\nb)}\phi(k)=d_1(\nb).
\end{align*}

\noindent\textbf{Case 1.2}  $\be(\nb)\neq 0$. Then  $\nb=2^{\be(\nb)}d'(\nb)d_1(\nb)$ and $r=d_1(r)$. Hence, by Lemma~\ref{BB},  we have
\begin{align*}
\B=\frac{1}{\phi(r)}\sum_{j\in\chi\cap M_q}\phi(j)=\frac{1}{\phi(r)}\sum_{k|\nb,~  kr\in M_q}\phi(kr)=\sum_{k| 2^{\be(\nb)}d_1(\nb)}\phi(k)=2^{\be(\nb)}d_1(\nb).
\end{align*}

\noindent \textbf{Case 2} $\gcd(\nb,r)\neq 1$.   Recall that $\nb=2^{\be(\nb)}\prod_{i=1}^sp_i^{a_i}\mu$ and $r=2^{\be(r)}\prod_{i=1}^sp_i^{b_i}\tau$ are factorizations of $\nb$ and $r$ as in Equation \eqref{nbr}.  By Lemma~\ref{lem4.7}, we have 
\begin{align*}
\B&=\frac{1}{\phi(r)}\sum_{j\in\chi\cap M_q}\phi(j)=\frac{1}{\phi(r)}\sum_{k|\mu,~ k\in M_q}\phi\left(r2^{\be(\nb)}\prod_{i=1}^s p_i^{a_i}k\right) \\
&=\frac{1}{\phi(r)}\sum_{k|\mu,~k\in M_q}\phi\left(2^{\be(\nb)+\be(r)}\prod_{i=1}^s p_i^{a_i+b_i}\tau k\right)\\
&=\frac{1}{\phi(r)}\cdot\phi(2^{\be(\nb)+\be(r)})\cdot\phi(\prod_{i=1}^s p_i^{a_i+b_i})\cdot\phi(\tau)\cdot\sum_{k| \frac{d_1(\nb)}{\prod_{i=1}^{s} {p_i}^{a_i}}}\phi(k)\\
&=\frac{1}{\phi(r)}\cdot\frac{2^{\be(\nb)+\be(r)}}{2}\cdot\prod_{i=1}^s p_i^{a_i+b_i}\cdot\left(1-\frac{1}{p_i}\right)\cdot\phi(\tau)\cdot\frac{d_1(\nb)}{\prod_{i=1}^{s} {p_i}^{a_i}}\\
&=\frac{1}{\phi(r)}\cdot\frac{2^{\be(\nb)+\be(r)}}{2}\cdot\prod_{i=1}^s p_i^{b_i}\cdot\left(1-\frac{1}{p_i}\right)\cdot\phi(\tau)\cdot{d_1(\nb)}\\
&=\frac{1}{\phi(r)}\cdot2^{\be(\nb)}\cdot \phi\left(2^{\be(r)}\right)\cdot \phi\left(\prod_{i=1}^s p_i^{b_i}\right)\cdot\phi(\tau)\cdot{d_1(\nb)}\\
&=2^{\be(\nb)}d_1(\nb).
\end{align*}

Next, assume that $r$ is  odd and $r\in M_q$. Then  $r=\prod_{i=1}^sp_i^{b_i}\tau=d_1(r)$.  By Lemma~\ref{lem4.7},  we have
\begin{align*}
\B=&~~\frac{1}{\phi(r)}\sum_{j\in\chi\cap M_q}\phi(j)=\frac{1}{\phi(r)}\sum_{k| 2^{\be(\nb)}\mu,~ k\in M_q}\phi\left(r\prod_{i=1}^sp_i^{a_i}k\right) \\ 
=&~~\frac{1}{\phi(r)}\sum_{k| 2^{\be(\nb)}\mu,~ k\in M_q}\phi\left(\prod_{i=1}^sp_i^{a_i+b_i}\tau k\right)\\
=&~~\frac{1}{\phi(r)}\cdot \phi\left(\prod_{i=1}^sp_i^{a_i+b_i}\right)\cdot\phi(\tau)\cdot\sum_{k\,\big\vert\, \frac{\nb}{\prod_{i=1}^s{p_i^{a_i}}},~k\in M_q} \phi(k)\\
=&~~\frac{1}{\phi(r)}\cdot \phi\left(\prod_{i=1}^sp_i^{a_i+b_i}\right)\cdot\phi(\tau)\cdot\sum_{k\,\big\vert\, 2^{\min\{\be(\nb),\gam\}}\frac{d_1(\nb)}{\prod_{i=1}^s{p_i^{a_i}}}} \phi(k)\\
=&~~\frac{1}{\phi(r)}\cdot \phi\left(\prod_{i=1}^sp_i^{a_i+b_i}\right)\cdot\phi(\tau)\cdot 2^{\min\{\be(\nb), \gam\}}\cdot \frac{d_1(\nb)}{\prod_{i=1}^s{p_i^{a_i}}}\\
=&~~\frac{1}{\phi(r)}\cdot\prod_{i=1}^sp_i^{a_i}\cdot \phi\left(\prod_{i=1}^sp_i^{b_i}\right)\cdot\phi(\tau)\cdot 2^{\min\{\be(\nb), \gam\}}\cdot \frac{d_1(\nb)}{\prod_{i=1}^s{p_i^{a_i}}}\\
=&~~2^{\min\{\be(\nb), \gam\}}d_1(\nb)
\end{align*}
as desired.
\end{proof}
By Lemma~\ref{lem48}  and Proposition~\ref{bnbq}, we obtain the following corollary.
\begin{cor}\label{bnn}
Let $\gamma\geq 0$ be an integer such that $2^\gam||(q+1)$. Then	$B_{\nb,\lambda,q^2}=\nb$ if and only if $r,~\nb\in M_q$ and $\be(\nb)+\be(r)\leq \gam$.
\end{cor}

\section{Bounds on $
    E_H(n,\lambda, q^2)$}

In this section, we focus on upper and lower bounds of $E_H(n,\lambda, q^2)$.  Based on this bounds, it can be concluded  that  either $E_H(n, \lambda, q^2) $ is $0$ or it  grows at the same rate with $n$ as $n$ tends to infinity.

\begin{thm}\label{thm:last}
	Let $q$ be a prime power and $2^\gam||(q+1)$. Let $n=\nb p^\nu$, where $p\nmid\nb$ and $\nu\geq 0$. Let $\lambda\in\Fqq^*$ be such that the multiplicative order of $\lambda$ is $r$ and $r| (q+1)$. Then one of the following statements holds.
	\begin{enumerate}
		\item $E_H(n,\lambda,q^2)=0$ if and only if $\be(\nb)+\be(r)\leq\gam$ and $r, n\in M_q$.
		\item If  $r$ is  even, then   one of the following statements holds.
        \begin{enumerate}
            \item  $\frac{n}{6}\leq E_H(n,\lambda,q^2)<\frac{n}{3}$  if $\be(\nb)+\be(r)\leq\gam$, $r\in M_q$ and $n\not\in M_q$.
            \item   $\frac{n}{4}\leq E_H(n,\lambda,q^2)<\frac{n}{3}$ if $\be(\nb)+\be(r)>\gam$ or $r\not\in M_q$.
        \end{enumerate} 
\item If  $r$ is  odd,  then   one of the following statements holds.
\begin{enumerate}
    \item  $\frac{n}{8}\leq E_H(n,\lambda,q^2)<\frac{n}{3}$ if  $r\in M_q$ and $n\not\in M_q$.
    \item   $\frac{n}{4}\leq E_H(n,\lambda,q^2)<\frac{n}{3}$ if $r\not\in M_q$.
    \end{enumerate}
		\end{enumerate}
\end{thm}

\begin{proof}
From Equation \eqref{EH},  observe that  $E_H(n,q^2)=0$ if and only if 
$$\frac{B_{{\nb,\lambda,q^2}}^H}{\nb}=\frac{4p^{2\nu}+2p^\nu}{p^{2\nu}+2p^\nu+3-3\delta_{p^\nu}}.$$
By Lemma~\ref{BB}, it is not difficult to see that $\frac{B_{{\nb,\lambda,q^2}}^H}{\nb}\leq 1$  and  $\frac{B_{{\nb,\lambda,q^2}}^H}{\nb}=1$ if and only if $r,~ \nb\in M_q$ and $\be(\nb)+\be(r)\leq\gam$ by Corollary~\ref{bnn}.  On the other hand,  we have 
$\frac{4p^{2\nu}+2p^\nu}{p^{2\nu}+2p^\nu+3-3\delta_{p^\nu}}\geq 1$ and  $\frac{4p^{2\nu}+2p^\nu}{p^{2\nu}+2p^\nu+3-3\delta_{p^\nu}}=1$ if and only if $p^{\nu}=1$.
Therefore,  $\frac{B_{{\nb,\lambda,q^2}}^H}{\nb}=\frac{4p^{2\nu}+2p^\nu}{p^{2\nu}+2p^\nu+3-3\delta_{p^\nu}}$ if and only if  they are $1$.  We conclude that $\frac{B_{{\nb,\lambda,q^2}}^H}{\nb}=\frac{4p^{2\nu}+2p^\nu}{p^{2\nu}+2p^\nu+3-3\delta_{p^\nu}}$ if and only if $r, \nb\in M_q$,~ $\be(\nb)+\be(r)\leq\gam$ and $p^{\nu}=1$. Equivalently,  $r, n\in M_q$ and $\be(\nb)+\be(r)\leq\gam$.
Therefore, Statement $1)$ is proved.	

In any cases, we have the upper bound $E_H(n,\lambda,q^2)<\frac{n}{3}$  from Corollary~\ref{coreq}.   Therefore, it remains to prove onlu the lower bounds.

To prove $2)$, assume that  $r$ is  even.  

$(a)$: Assume  that  $\be(\nb)+\be(r)\leq\gam$ , $r\in M_q$, and $n\not\in M_q$.

\noindent \textbf{Case 1} $\gcd(n,q)\neq 1$.  
	By  Corollary \ref{coreq}, we have $E_H(\nb,\lambda, q^2)=\frac{\nb-B_{{\nb,\lambda, q^2}}^H}{4}$. Hence, by Equation \eqref{EH}, $E_H(n, \lambda, q^2)$ can be  expressed as 
\begin{align}\notag
	\frac{E_H(n, \lambda, q^2)}{n}=&\frac{1}{4}-\frac{1}{4p^\nu}+\frac{\delta_{p^\nu}}{4p^\nu(p^\nu+1)}+\frac{E_H(\nb, \lambda, q^2)}{n}\left(\frac{p^\nu+1}{3}+\frac{2-3\delta_{p^\nu}}{3(p^\nu+1)}\right)\\\notag
	\geq&\frac{1}{4}-\frac{1}{4p^\nu}+\frac{\delta_{p^\nu}}{4p^\nu(p^\nu+1)}.
\end{align}
It follows that $\frac{E_H(n, q^2)}{n}\geq\frac{1}{6}$ for all $p^\nu\geq 2$. Hence, ${E_H(n, q^2)}\geq\frac{n}{6}$.
\noindent \textbf{Case 2} $\gcd(n,q)=1$. Then $n=\nb=2^{\be(\nb)}d'(\nb)d_1(\nb)$. By Proposition~\ref{bnbq}, we have $\B=2^{\be(\nb)}d_1(\nb)$. It follows that 
\begin{align*}
\frac{\B}{n}=\frac{2^{\be(\nb)}d_1(\nb)}{2^{\be(\nb)}d'(\nb)d_1(\nb)}=\frac{1}{d'(\nb)},
\end{align*}
and hence,  $\B=\frac{n}{d'(\nb)}$. Therefore, 
\begin{align*}
E_H(n,\lambda, q^2)=\frac{n-\B}{4}=\frac{1}{4}\left(n-\frac{n}{d'(\nb)}\right)=\frac{n}{4}\left(1-\frac{1}{d'(\nb)}\right)\geq \frac{\nb}{4}\left(1-\frac{1}{3}\right)=\frac{n}{6}
\end{align*}
Altogether, it can be concluded  that $\frac{n}{6}\leq E_H(n,\lambda,q^2)$. 

$(b)$: Assume that  $\be(\nb)+\be(r)>\gam$ or $r\not\in M_q$. By Corollary \ref{cor-noB}, we have  \[E_H(n,\lambda,q^2)=n\left(\frac{1}{3}-\frac{1}{6(p^\nu+1)}\right)\geq n\left(\frac{1}{3}-\frac{1}{6(1+1)}\right)=\frac{n}{4}.\]

To prove $3)$, assume that  $r$ is odd.

$(a)$: Assume  that  $r\in M_q$ and $n\not\in M_q$. 

\noindent \textbf{Case 1} $\gcd(n,q)\neq 1$. 
Similar to the prove of Case 1 in 2), we have ${E_H(n, q^2)}\geq\frac{n}{6}$.

\noindent \textbf{Case 2} $\gcd(n,q)=1$. 
Then $\B=2^{\min\{\gam,~ \be(\nb)\}}d_1(\nb)$  by Proposition~\ref{bnbq}. Thus
\begin{align*}
\frac{\B}{n}=\frac{2^{\min\{\gam,~ \be(\nb)\}}d_1(\nb)}{2^{\be(\nb)}d'(\nb)d_1(\nb)}=\frac{1}{2^{\be(\nb)-\min\{\gam,~ \be(\nb)\}}d'(\nb)}
\end{align*}
Since $n\not\in M_q$,  we  consider  the  following   $2$ cases.  

\noindent \textbf{Case 2.1} $d'(\nb)=1$ and  $\be(\nb)>\gam$.  We have $\frac{B_{{\nb,\lambda,q^2}}^H}{n}=\frac{1}{2^{\be(\nb)-\gam}}\leq \frac{1}{2}$.

\noindent \textbf{Case 2.2} $d'(\nb)>1$.   Then  \[\frac{B_{{\nb,\lambda,q^2}}^H}{n}=\frac{1}{2^{\be(\nb)-\min\{\gam,~\be(\nb)\}}d'(\nb)}\leq \frac{1}{d'(\nb)}\leq \frac{1}{3}.\]
From Cases 2.1 and 2.2, we get $B_{\nb, q^2}^H\leq \frac{n}{2}$ which implies that $$E_H(n, \lambda, q^2)=\frac{n-B_{{\nb,\lambda,q^2}}^H}{4}\geq \frac{n-\frac{n}{2}}{4}=\frac{n}{8}.$$
Consequently, we have $\frac{n}{8}\leq E_H(n, \lambda, q^2)$.

$(b)$: Assume that  $r\not\in M_q$. By Corollary \ref{cor-noB}, we have  \[E_H(n,\lambda,q^2)=n\left(\frac{1}{3}-\frac{1}{6(p^\nu+1)}\right)\geq n\left(\frac{1}{3}-\frac{1}{6(1+1)}\right)=\frac{n}{4}.\]

The proof is completed.\end{proof}

\begin{rem} \label{rem:last} Assume the notations as in Theorem~\ref{thm:last}.  If $n\in M_q$, then $\nb \in M_q$ by  Proposition \ref{ell} which implies that  $\be(\nb) \leq\gam$.    In the case where  $r$ is odd, we have   $\be(r)=0$ which implies that  $\be(\nb)+\be(r)=\be(\nb)\leq\gam$. 
  Hence,  if  $r$ is odd,  it can be concluded that  $E_H(n,\lambda,q^2)=0$ if and only if  $r, n\in M_q$.
\end{rem}

From Theorem~\ref{thm:last}, it can be concluded that either  $E_H(n, \lambda, q^2) =0$ or it  grows at the same rate with $n$ as $n$ tends to infinity.

 \section{Conclusion and Remarks}
  For an element $\lambda$ in $\Fqq$ of order $r$ such that $r|(q+1)$, a general formula for the average dimension $E_H(n, \lambda, q^2) $ of the hull of $\lambda$-constacyclic codes has been given in Theorem \ref{thm1} and as well as its lower and upper bounds in Theorem  \ref{thm:last}.
  Asymptotically,   either the  average dimension of the Hermitian hull of constacyclic  codes is    zero or it grows at the same rate with $n$. 
  
  Let   $E(n,  q)$   denote the average dimension of the Euclidean hull of cyclic codes of length $n$ over $\Fq$ and let $N_q:=\{\ell\geq 1| \ell \text{~divides~} q^i+1 \text{~for some positive integer~} i\}$. In {\cite[Theorem 25]{Skersys03}}, it has been shown that  $E(n, q)=0$   if and only if  $n\in N_q$, and $\frac{n}{12}\leq E(n,  q)<\frac{n}{3}$ otherwise.  This means that  either  the  average dimension of the Euclidean  hull of cyclic  codes is    zero or  it grows at the same rate with $n$. This result coincides  the result for Hermitian case in this paper.  However,   there are interesting  difference on  lower bounds  explained in Table \ref{T1}.



  
  \begin{table}[!hbt]
      \centering
      \begin{tabular}{|c|c|c|c|l|}
         \hline Order of $\lambda$ &$n$ &LB&UB& Remarks\\
         \hline
 \multirow{ 2}{*}{$r$ is odd}& $r\in M_q$ and $ n\in M_q$ &$0$&$0$&   Remark \ref{rem:last} \\\cline{2-5}
 \multirow{ 2}{*}{and $r|(q+1).$}& $r\in M_q$ and $ n\not\in M_q$&$\dfrac{n}{8}$&$\dfrac{n}{3}$&  \multirow{ 2}{*}{Theorem~\ref{thm:last} }\\\cline{2-4}
& $r\not\in M_q$&$\dfrac{n}{4}$&$\dfrac{n}{3}$&\\ \cline{1-5}
\multirow{ 5}{*}{$r$ is even}& $\be(\nb)+\be(r)\leq\gam$, &$0$&$0$& \multirow{ 6}{*}{Theorem~\ref{thm:last}}\\     
 & $r\in M_q$ and $n\in M_q$&&&\\   \cline{2-4}  
 & $\be(\nb)+\be(r)\leq\gam$, &$\dfrac{n}{6}$&$\dfrac{n}{3}$&  \\
\multirow{ 1}{*}{and $r|(q+1).$}&$r\in M_q$ and  $n\not\in M_q$&&&\\ \cline{2-4}
& $\be(\nb)+\be(r)>\gam$ &$\dfrac{n}{4}$&$\dfrac{n}{3}$&\\
&or $r\not\in M_q$&&&\\
         \hline
         
      \end{tabular}
  \caption{The lower and upper bounds for  $E_H(n, \lambda, q^2)$}
  \label{T1}
  \end{table}

 From the summary in Table \ref{T1}, it would be  interesting to study an improvement on the lower and upper bounds  for  the average dimension of the hull of constacyclic codes with some restricted lengths.

\end{document}